\documentclass[11pt,leqno]{article}

\usepackage[margin=1in]{geometry}

\usepackage[font=small,skip=0pt]{caption}

\usepackage[hyphens]{url}

\usepackage[style=alphabetic,natbib=true,maxalphanames=10,maxnames=10,sortcites=false,doi=false,url=false]{biblatex} %

\usepackage{microtype}          %
\usepackage[T1]{fontenc}

\usepackage{lmodern}
\usepackage[full]{textcomp}     %

\usepackage[citecolor=magenta,anchorcolor=blue,colorlinks]{hyperref}%

\usepackage{amsmath}

\usepackage{cleveref}
\usepackage{article} %
\usepackage{math} %
\usepackage{algo} %
\usepackage{wrapfig} %
\usepackage{placeins} %
\usepackage{xcolor}

\newcommand{\gain}{\delta\optsub} %
\newcommand{\Sout}{\hat{S}} %
\newcommand{\exitvalue}{\chi\optpar} %
\newcommand{\incv}[3]{\nu\parof{#1,#2,#3}} %
\newcommand{\incvalue}{\incv} %

\newcommand{\schedulebefore}{S^-\optsub} %
\newcommand{\scheduleafter}{S^+\optsub} %
\newcommand{\setbefore}{\schedulebefore} 
\newcommand{\setafter}{\scheduleafter} %

\newcommand{\edge}{} %

\newcommand{\crs}{{\sc CR}\xspace}
\newcommand{\cI}{\mathcal{I}}

\providecommand{\Sout}{\smash{\hat{S}}}%
\providecommand{\Stack}{S}%
\providecommand{\Stackend}{S_{\text{end}}}

\providecommand{\afterneighbors}[1]{A_{#1}}
\providecommand{\beforeneighbors}[1]{B_{#1}}

\NewDocumentEnvironment{Inset}{}{
  \setlength{\topsep}{0pt}
  \fontsize{10pt}{12pt}\selectfont
  \renewcommand{\paragraph}[1]{

    \medskip\noindent\textbf{##1}}%
  \providecommand{\subparagraph}[1]{

    \medskip\noindent\textit{##1}}%
  \renewcommand{\subparagraph}[1]{

    \medskip\noindent\textit{##1}}%

  \noindent%
  \begin{oframed}%
  }{
  \end{oframed}%
  \medskip
}

\renewcommand{\paragraph}[1]{\medskip \noindent\textbf{#1}\xspace} %

\bibliography{elimination-graphs}

\title{Independent Sets in Elimination Graphs \\ with a Submodular
  Objective}                    %

\author{ Chandra Chekuri\thanks{Dept.\ of Computer
    Science, University of Illinois at Urbana-Champaign. {\tt
      chekuri@illinois.edu}. Supported in part by NSF
    grants CCF-1910149 and CCF-1907937.}
  \and Kent
  Quanrud\thanks{Dept.\ of Computer Science, Purdue University, West Lafayette,
    IN. {\tt krq@purdue.edu}.  Supported in part by NSF grant CCF-2129816.}
}  %

\begin{document}

\maketitle

\begin{abstract}
  Maximum weight independent set (MWIS) admits a
  $\frac1k$-approximation in inductively $k$-independent graphs
  \cite{AkcogluADK02,YeB12} and a $\frac{1}{2k}$-approximation in
  $k$-perfectly orientable graphs \cite{KammerT14}.  These are a
  parameterized class of graphs that generalize $k$-degenerate graphs,
  chordal graphs, and intersection graphs of various geometric shapes
  such as intervals, pseudo-disks, and several others \cite{YeB12,
    KammerT14}.  We consider a generalization of MWIS to a submodular
  objective.  Given a graph $G=(V,E)$ and a non-negative submodular
  function $f: 2^V \rightarrow \mathbb{R}_+$, the goal is to
  approximately solve $\max_{S \in \mathcal{I}_G} f(S)$ where
  $\mathcal{I}_G$ is the set of independent sets of $G$.  We obtain an
  $\Omega(\frac1k)$-approximation for this problem in the two
  mentioned graph classes. The first approach is via the multilinear
  relaxation framework and a simple contention resolution scheme, and
  this results in a randomized algorithm with approximation ratio at
  least $\frac{1}{e(k+1)}$. This approach also yields parallel (or
  low-adaptivity) approximations.

  Motivated by the goal of designing efficient and deterministic
  algorithms, we describe two other algorithms for inductively
  $k$-independent graphs that are inspired by work on streaming
  algorithms: a preemptive greedy algorithm and a primal-dual
  algorithm. In addition to being simpler and faster, these
  algorithms, in the monotone submodular case, yield the first
  deterministic constant factor approximations for various special
  cases that have been previously considered such as intersection
  graphs of intervals, disks and pseudo-disks.
\end{abstract}

\section{Introduction}
Given a graph $G=(V,E)$ a set $S \subseteq V$ of vertices is an
independent set (also referred to as a stable set) if there is no edge
between any two vertices in $S$.  Let $\alpha(G)$ denote the
cardinality of a maximum independent set in $G$.  Finding $\alpha(G)$
is a classical problem with many applications; we refer to the search
problem of finding a maximum cardinality independent set as MIS. We
also consider the weighted version where the input consists of $G$ and
a vertex weight function $w: V \rightarrow \mathbb{Z}_+$ and the goal
is to find a maximum weight independent set; we refer to the weighted
problem as MWIS.  MIS is NP-Hard, and moreover it is also NP-Hard to
approximate $\alpha(G)$ to within a $\frac{1}{n^{1-\eps}}$-factor for
any fixed $\eps > 0$ \cite{Hastad99,Zuckerman06}. For this reason, MIS
and MWIS are studied in various special classes of graphs that capture
interesting problems while also being tractable. It is easy to see
that graphs with maximum degree $k$ admit a $\frac{1}{k}$-approximation.
In fact, the same approximation ratio holds for
$k$-degenerate graphs --- a graph $G=(V,E)$ is a $k$-degenerate if there is an
ordering of the vertices $\vertices = \setof{v_1,\dots,v_n}$ such that
for each $v_i$, $|N(v_i) \cap \setof{v_i,\dots,v_n}| \le k$. A canonical
example is the class of planar graphs which are $5$-degenerate.

In this paper we are interested in two parameterized classes of graphs
called inductively $k$-independent graphs \cite{YeB12} and
$k$-perfectly orientable graphs \cite{KammerT14}. These graphs are
motivated by the well-known class of chordal graphs, and capture
several other interesting classes such as intersection graphs of
intervals, disks (and hence planar graphs), low-treewidth graphs,
$t$-interval graphs, and many others.  A more recent example is the
intersection graph of a collection of pseudo-disks which were shown to
be inductively $156$-independent \cite{Pinchasi14}.  Graphs in these
classes can be dense and have large cliques. We formally define the
classes.

Given a graph $G=(V,E)$ and a vertex $v$ we let $N(v)$ denote the set
of neighbors of $v$ (excluding $v$).  A graph $\defgraph$ with $n$
vertices has a \emph{perfect elimination ordering} if there is an
ordering of vertices $\vertices = \setof{v_1,\dots,v_n}$ such that for
each $v_i$, $\alpha(G[N(v_i) \cap \setof{v_i,\dots,v_n}]) = 1$; in
other words $N(v_i) \cap \setof{v_i,\dots,v_n}$ is a clique. It is
well-known that these graphs are the same as chordal
graphs.\footnote{A graph is chordal iff there is no induced cycle of
  length more than $3$.} For example, the intersection graph of a
given set of intervals is chordal. One can generalize the perfect
elimination property ordering of chordal graphs.

\begin{definition}[\cite{KammerT14}]
  For a fixed integer $k \ge 1$, $\defgraph$ is \emph{$k$-simplicial} if
  there is an ordering of vertices $\vertices = \setof{v_1,\dots,v_n}$
  such that for each $v_i$, $G[N(v_i) \cap \setof{v_i,\dots,v_n}]$
  can be covered by $k$ cliques.
\end{definition}

Note that if $G[N(v_i) \cap \setof{v_i,\dots,v_n}]$ is covered by $k$
cliques then $\alpha(G[N(v_i) \cap \setof{v_i,\dots,v_n}]) \le
k$. Hence one can define a class based on this weaker property.

\begin{definition}[\cite{AkcogluADK02,YeB12}]
  For a fixed integer $k \ge 1$, $\defgraph$ is \emph{inductively
    $k$-independent} if there is an ordering of vertices $\vertices =
  \setof{v_1,\dots,v_n}$ such that for each $v_i$, $\alpha(G[N(v_i) \cap
  \setof{v_i,\dots,v_n}]) \le k$.  The \emph{inductive independence
    number} of $\graph$ is the minimum $k$ for which $\graph$ is
  inductively $k$-independent.
\end{definition}

Although inductively $k$-independent graphs generalize $k$-simplicial
graphs there is no known natural class of graphs that differentiates
the two; typically one establishes inductive $k$-independence via
$k$-simpliciality. The ordering-based definition can be further
relaxed based on orientations of $G$.

\begin{definition}[\cite{KammerT14}]
  For a fixed integer $k \ge 1$, $\defgraph$ is \emph{$k$-perfectly
    orientable} if there is an orientation $H=(V,A)$ of $\graph$ such
  that for each vertex $v \in \vertices$, $G[S_v]$ can be covered by
  $k$ cliques, where $S_v = N^+_H(v)$ is the out-neighborhood of $v$
  in $H$.
\end{definition}

\begin{remark}
  In this paper we will use the term $k$-perfectly orientable for the
  following class of graphs: there is an orientation $H=(V,A)$ of
  $\graph$ such that for each vertex $v \in \vertices$,
  $\alpha(G[S_v]) \le k$ where $S_v = N^+_H(v)$ is the
  out-neighborhood of $v$ in $H$.  This is more general than the
  preceding definition. We observe that the algorithm in
  \cite{KammerT14} for MWIS works also for this larger class, although
  there are no known natural examples that differentiate the two.
\end{remark}

We observe that if $G$ is inductively $k$-independent then it is also
$k$-perfectly orientable according to our relaxed definition. Indeed,
if $v_1,v_2,\ldots,v_n$ is an ordering that certifies inductive
$k$-independence we simply orient the edges of $G$ according to this
ordering which yields a DAG. The advantage of the $k$-perfect
orientability is that it allows arbitrary orientations. Note that a
cycle is $1$-perfectly orientable while it is $2$-inductively
independent. This factor of $2$ gap shows up in the known
approximation bounds for MWIS in these two classes of graphs. It is
known that for arbitrarily large $n$ there are $2$-perfectly
orientable graphs on $n$ vertices such that the graphs are not
inductively $\sqrt{n}$-independent \cite{BaryehudaHNSS06}.  These come
from the intersection graphs of so-called $2$-interval graphs.  Thus,
$k$-perfect orientability can add substantial modeling power.

Akcoglu et al.\ \cite{AkcogluADK02} described a
$\frac1k$-approximation for the MWIS problem in graphs that are
inductively $k$-independent.  They used the local-ratio technique, and
subsequently \cite{YeB12} derived it using a stack-based
algorithm. Both algorithms require as input an ordering of the
vertices that certifies the inductive $k$-independent property. For
$k$-perfectly orientable graphs \cite{KammerT14} described a
$\frac{1}{2k}$-approximation for the MWIS problem following the ideas
in \cite{BaryehudaHNSS06} for a special case.  Given a graph
$G=(V,E)$ and integer $k$ there is an $n^{O(k)}$-time algorithm to
check if $G$ is inductively $k$-independent \cite{YeB12}. Typically, the proof that a
specific class of graphs is inductively $k$-independent for some fixed
value of $k$, yields an efficient algorithm that also computes a
corresponding ordering. This is also true for $k$-perfect
orientability.  We refer the reader to \cite{HalldorssonT21} for
additional discussion on computational aspects of computing
orderings. In this paper we will assume that we are given both $G$ and
the ordering that certifies inductive $k$-independence, or an
orientation that certifies $k$-perfect orientability.

\subsection{Independent sets with a submodular objective}
We consider
an extension of MWIS to submodular objectives.  A real-valued set
function $f:2^V \rightarrow \mathbb{R}$ is modular iff
$f(A) + f(B) = f(A\cup B) + f(A \cap B)$ for all $A, B \subseteq V$.
It is easy to show that $f$ is modular iff there a weight function
$w: V \rightarrow \mathbb{R}$ where
$f(A) = w(A) = \sum_{v \in A} w(v)$.  A real-valued set function
$f:2^V \rightarrow \mathbb{R}$ is \emph{submodular} if
$f(A) + f(B) \ge f(A \cup B) + f(A \cap B)$ for all
$A, B \subseteq V$.  An equivalent definition is via decreasing
marginal value property: for any $A \subset B \subset V$ and
$v \in V - B$, $f(A+v) - f(A) \ge f(B+v) - f(B)$. Here $A+v$ is a
convenient notation for $A \cup \{v\}$. $f$ is monotone if
$f(A) \le f(B)$ for all $A \subseteq B$.  We will confine our
attention in this paper to non-negative submodular functions and we
will also assume that $f(\emptyset) = 0$.  Given a graph $G=(V,E)$ and
a non-negative submodular function $f:2^V \rightarrow \mathbb{R}_+$,
we consider the problem $\max_{S \subseteq \mathcal{I}_G} f(S)$ where
$\mathcal{I}_G$ is the collection of independent sets in $G$. This
problem generalizes MWIS since a modular function is also
submodular. We assume throughout that $f$ is available through a value
oracle that returns $f(S)$ on query $S$.  Our focus is on developing
approximation algorithms for this problem in the preceding graph classes,
since even very simple special cases are NP-Hard.

\medskip
\noindent
\emph{Motivation and related work:} Submodular function maximization
subject to various ``independence'' constraints has been a very active
area of research in the last two decades. There have been several
important theoretical developments, and a variety of applications
ranging from algorithmic game theory, machine learning and artificial
intelligence, data analysis, and network analysis; see
\cite{BuchbinderF-survey,Bilmes22,Clarketal16} for some pointers.  We
are motivated to consider this objective in inductive $k$-independent
graphs and $k$-perfectly orientable graphs for several reasons. First,
it is a natural generalization of MWIS. Second, various special cases
of this problem have been previously studied: Feldman
\cite{Feldman-thesis} considered the case of interval graphs, and Chan
and Har-Peled considered the case of intersection graphs of disks and
pseudo-disks \cite{ChanH12}. Third, previous algorithms have relied on
the multilinear relaxation based approach combined with contention
resolution schemes for rounding.  This is a computationally expensive
approach and also requires randomization.  The known approximation
algorithms for MWIS in inductive $k$-independent graphs are based on
simple combinatorial methods such as local-ratio, and this raises the
question of developing similar combinatorial algorithms for submodular
objectives. In particular, we are inspired by the connection to
\emph{preemptive} greedy algorithms for submodular function
maximization that have been developed in the context of streaming
algorithms \cite{ChakrabartiK15,bmkk-sso-14,ChekuriGQ15}. Although a
natural greedy algorithm has been extensively studied for submodular
function maximization \cite{nwf-mssf1-78,fnw-mssf2-78}, the utility of
the preemptive version for offline approximation has not been explored
as far as we are aware of. This is partly due to the fact that the
standard greedy algorithm works well for matroid like constraints.
More recently \cite{LevinW21} developed a primal-dual based algorithm
for submodular streaming under $b$-matching constraints which is
inspired by the stack based algorithm of \cite{PazS18} for the modular
setting; the latter has close connections to stack based algorithms
for inductive $k$-independent graphs \cite{YeB12}.  The algorithm in
\cite{LevinW21} was generalized to matroid intersection in
\cite{GargJS22}.  Finally, at a meta-level, we are also interested in
understanding the relationship in approximability between optimizing
with modular objectives and submodular objectives.  For many
``independence'' constraints the approximability of the problem with a
submodular objective is often within a constant factor of the
approximability with a modular objective, but there are also settings
in which the submodular objective is provably harder (see \cite{BruggmannZ19}).  A substantial
amount of research on submodular function optimization is for
constraints defined by \emph{exchange} systems such as (intersections
of) matroids and their generalizations such as $k$-exchange systems
\cite{FeldmanNSW11} and $k$-systems
\cite{Jenkyns76,CCPV11}. Independent sets in the graph classes we
consider provide a different parameterized family of constraints.

\subsection{Results}
We obtain an $\Omega(\frac1k)$-approximation for $\max_{S \subseteq
  \mathcal{I}} f(S)$ in inductively $k$-independent graphs and in
$k$-perfectly orientable graphs. We explore different techniques to
achieve these results since they have different algorithmic benefits.

First, we obtain a randomized algorithm via the multilinear relaxation framework
\cite{CVZ14-crs} by considering a natural polyhedral relaxation and
developing simple contention resolution schemes (CRS). The CRS schemes
are useful since one can combine the rounding with other side
constraints in various applications.

\begin{theorem}
  \label{thm:intro-crs}
  There is a randomized algorithm that given a $k$-perfectly orientable graph $G$
  (along with its orientation) and a monotone submodular function $f$, outputs
  an independent set $S'$ such that with high probability
  $f(S') \ge (\frac{1}{k+1}
  \cdot \frac{1}{(1+1/k)^k}) \max_{A \in \cI_G} f(A)$. For non-negative functions
  there is an algorithm that outputs an independent set $S'$ such that
  with high probability $f(S') \ge \frac{1}{e(k+1)} \max_{A \in \cI_G} f(A)$.
\end{theorem}

The multilinear relaxation based approach yields parallel (or
low-adaptivity) algorithms with essentially similar approximation
ratios, following ideas in \cite{ChekuriQ19,EneNV19}.  Although the
multilinear approach is general and powerful, there are two drawbacks;
algorithmic complexity and randomization which are inherent to the
approach. An interesting question in the submodular maximization
literature is whether one can obtain deterministic algorithms via
alternate methods, or by derandomizing the multilinear relaxation
approach. There have been several results along these lines
\cite{BuchbinderF18,BuchbinderF19,HanCCW20}, and several open
problems.

Motivated by these considerations we develop simple and efficient
approximation algorithms for inductively $k$-independent
graphs. We show that a preemptive greedy algorithm, inspired by the
streaming algorithm in \cite{ChekuriGQ15}, yields a
deterministic $\Omega(\frac1k)$-approximation when $f$ is monotone.  This can be
combined with a simple randomized approach when $f$ is non-monotone.
Inspired by \cite{LevinW21}, we describe a primal-dual
algorithm that also yields a $\Omega(\frac1k)$-approximation; the primal-dual
approach yields better constants and we state the result below.

\begin{theorem}
  \label{thm:intro-primal-dual}
  There is a deterministic combinatorial algorithm that given an inductively
  $k$-independent graph $G$ (along with its orientation) and a
  monotone submodular function $f$, outputs an independent set $S'$
  such that
  $f(S') \ge \frac{1}{k+1+2 \sqrt{k}} \max_{A \in \cI_G} f(A)$. For
  non-negative functions there is a randomized algorithm that outputs
  an independent set $S'$ such that
  $\evof{f(S')} \ge \frac{1}{2k + 1 +\sqrt{8k}} \max_{A \in \cI_G}
  f(A)$. Both algorithms use $O(|V(G)|)$ value oracle calls to $f$ and
  in addition take linear time in the size of $G$.
\end{theorem}

\begin{remark}
  We obtain deterministic $1/4$-approximation for
  monotone submodular function maximization for independent sets in
  chordal graphs, and hence also for interval graphs. This matches the best ratio
  known via the multilinear relaxation approach
  \cite{Feldman-thesis}, and is the first deterministic algorithm as
  far as we know. Similarly, this is the first deterministic
  algorithm for disks and pseudo-disks that were previously handled via
  the multilinear relaxation approach \cite{ChanH12}. Are there deterministic
  algorithms for $k$-perfectly orientable graphs?  See \Cref{sec:concl}.
\end{remark}

\begin{remark}
  Matchings in a graph $G$, when viewed as independent sets in the
  line graph $H$ of $G$, form an inductively $2$-independent graph. In
  fact \emph{any} ordering of the edges of $G$ forms a valid
  $2$-inductive ordering of $H$. Thus our algorithm is also
  a semi-streaming algorithm. Our approximation bound for monotone
  functions matches the approximation achieved in \cite{LevinW21} for
  matchings although we use a different LP relaxation and view the
  problem from a more general viewpoint. However, for non-monotone
  functions, our ratio is slightly weaker, and highlights some
  differences.
\end{remark}

The primal-dual algorithm is a two-phase algorithm. The preemptive
greedy algorithm is a single phase algorithm. It gives slightly weaker
approximation bounds when compared to the primal-dual algorithm, but
has the advantage that it can be viewed as an \emph{online preemptive}
algorithm. Algorithms in such a model for submodular maximization were
developed in \cite{BuchbinderFS19,FeldmanKK18}.  Streaming
algorithms for submodular function maximization in
\cite{ChakrabartiK15,ChekuriGQ15} can be viewed as online preemptive
algorithms. Our work shows that there is an online preemptive
algorithm for independent sets of inductive $k$-independent graphs if
the vertices arrive in the proper order. There
are interesting examples where any ordering of the vertices is a valid
$k$-inductive ordering.

Our main contribution in this paper is conceptual. We study the
problem to unify and generalize existing results,  understand the
limits of existing techniques, and raise some directions for future
research (see \Cref{sec:concl}). As we mentioned, our
techniques are inspired by past and recent work on submodular function
maximization \cite{CVZ14-crs,Feldman-thesis,ChekuriGQ15,LevinW21}.

\paragraph{Organization:} \Cref{sec:prelim} sets up the relevant
technical background on submodular functions. \Cref{sec:crs} describes
the multilinear relaxation approach and proves \Cref{thm:intro-crs}.
\Cref{sec:primal-dual} describes the primal-dual approach and proves
Theorem~\ref{thm:intro-primal-dual}. \Cref{sec:greedy}
described the preemptive greedy algorithm. \Cref{sec:concl} concludes
with a discussion of a few future directions.

\section{Preliminaries}
\label{sec:prelim}
\paragraph{Submodular functions.}
Let $f: \subsetsof{\groundset} \to \nnreals$ be a nonnegative set
function defined over a ground set $\groundset$. ($\groundset$ may
have infinite cardinality, in which case it suffices that $f$ is
defined over the finite subsets of $\groundset$.) The function $f$ is
\emph{monotone} if $f(S) \leq f(T)$ for any nested sets
$S \subseteq T \subseteq \groundset$, and \emph{submodular} if it has
decreasing marginal returns: if $S \subseteq T \subseteq \groundset$
are two nested sets and $e \in \groundset$ is another element, then
\begin{math}
  f(S + e) - f(S) \geq f(T + e) - f(T).
\end{math}
For two sets $A,B \subseteq \groundset$, we denote the marginal value
of adding $B$ to $A$ by
\begin{math}
  f_A(B) \defeq f(A \cup B) - f(A).
\end{math}

\paragraph{Incremental values.}
In this paper, there is always an implicit ordering $<$ over the
ground set $\groundset$. For a set $S \subseteq \groundset$ and an
element $e \in \groundset$, the \emph{incremental value} of $e$ in
$S$, denoted $\incv{f}{S}{e}$, is defined as
\begin{align*}
  \incv{f}{S}{e}                %
  =                             %
  f_{S'}(e), \text{ where } S'  %
  =                             %
  \setof{s \in S : s < e}.
\end{align*}
Incremental value has some simple but very useful properties, proved
in \cite[Lemmas 1--3]{ChekuriGQ15} and summarized in the following.
\begin{lemma}
  Let $\groundset$ be an ordered set and
  $f: \subsetsof{\groundset} \to \reals$ a set function.
  \begin{mathproperties}
  \item For any set $S \subseteq \groundset$, we have
    \begin{math}
      f(S) = \sum_{e \in S} \incv{f}{S}{e}.
    \end{math}
  \item Let $S \subseteq T \subseteq \groundset$ be two nested subsets
    of $\groundset$ and $e \in \groundset$ an element. If $f$ is
    submodular, then
    \begin{math}
      \incv{f}{T}{e} \leq \incv{f}{S}{e}.
    \end{math}
  \item Let $S,Z \subseteq \groundset$ be two sets, and $e \in S$. If
    $f$ is submodular, then
    \begin{math}
      \incv{f_Z}{S}{e} \leq \incv{f}{Z \cup S}{e}.
    \end{math}
  \end{mathproperties}
\end{lemma}

\paragraph{Multilinear Extension and Relaxation.}

\begin{definition}
  Given a set function $f: \groundsets \to \reals$, the
  \emph{multilinear extension} of $f$, denoted $F$, extends $f$
  to the product space $[0,1]^{\groundset}$ by interpreting each point
  $x \in [0,1]^{\groundset}$ as an independent sample
  $S \subseteq \groundset$ with sampling probabilities given by $x$,
  and taking the expectation of $f(S)$. Equivalently,
  $$F(x) = \sum_{S \subseteq \groundset} \parof{\prod_{i \in S}x_i
    \prod_{i \not \in S} (1-x_i)}.$$
\end{definition}

An \emph{independence family} $\mathcal{I}$ over a ground set
$\groundset$ is a subset of $2^\groundset$ that is downward closed,
that is, if $A \in \mathcal{I}$ and $B \subset A$ then
$B \in \mathcal{I}$.  A polyhedral/convex relaxation $P$ for a given
independence family $\mathcal{I}$ over $\groundset$ is a
polyhedra/convex subset of $[0,1]^\groundset$ such that for each
$A \in \mathcal{I}$, $\chi_A \in P$ where $\chi_A$ is the
characteristic vector of $A$ (a vector in $\{0,1\}^\groundset$ with a
$1$ in coordinate $i$ iff $i \in A$). We say that $P$ is a
\emph{solvable} relaxation for $\mathcal{I}$ if there is a polynomial
time algorithm to optimize a linear objective over $P$. Given a
ground set $\groundset$, and a non-negative submodular function $f$
over $\groundset$, and an independence family
$\mathcal{I} \subseteq 2^\groundset$,\footnote{We assume that an
  independence family is specified implicitly via an independence
  oracle that returns whether a given $A \subseteq \groundset$ belongs
  to $\mathcal{I}$.} we are interested in the problem
$\max_{S \in \mathcal{I}} f(S)$. For this general problem the
multilinear relaxation approach is to approximately solve the
multilinear relaxation $\max_{x \in P} F(x)$ followed by rounding ---
see \cite{CCPV11,CVZ14-crs,BuchbinderF-survey}.  For monotone $f$
there is a randomized $(1-1/e)$-approximation to the multilinear
relaxation when $P$ is solvable \cite{CCPV11}. For general
non-negative functions there is a $0.385$-approximation
\cite{BuchbinderF19}.

\paragraph{Concave closure and relaxation.}

\begin{definition}
  Given a set function $f: \groundsets \to \reals$, the
  \emph{concave closure} of $f$, denoted $f^+$, extends $f$
  to the product space $[0,1]^{\groundset}$ as follows.
  For $x \in [0,1]^{\groundset}$ we let
  \begin{align*}
    f^+(x) = \max\setof{ \sum_{S \subseteq \groundset} \alpha_S f(S):
    \sum_{S \ni i} \alpha_S = x_i \text{ for all } i \in \groundset \andcomma \sum_{S}
    \alpha_S = 1 \andcomma \alpha_S \ge
    0 \text{ for all } S \subseteq N}.
  \end{align*}
\end{definition}
As the name suggests, $f^+$ is a concave function over
$[0,1]^\groundset$ for any set function $f$. The definition of
$f^+(x)$ involves the solution of an exponential sized linear program.
The concave closure of a submodular set function is in general NP-Hard
to evaluate. Nevertheless, the concave closure is useful indirectly in
several ways. One can relate the concave closure to the multilinear
extension via the notion of correlation gap
\cite{AgrawalDSY10,AgrawalDSY12,ccpv-07,Vondrak-thesis,ChekuriL21}.
We can consider a
relaxation based on the concave closure for the problem of
$\max_{S \in \cI} f(S)$, namely, $\max_{x \in P} f^+(x)$ where $P$ is
a polyhedral or convex relaxation for the constraint set $\cI$.
Although we may not be able to solve this relaxation directly, it
provides an upper bound on the optimum solution and moreover, unlike
the multilinear relaxation, the relaxation can be rewritten as a large
linear program when $P$ is polyhedral.

\paragraph{Contention Resolution Schemes.} Contention resolution schemes
are a way to round fractional solutions for relaxations to packing problems
and they are a powerful and useful tool in submodular function maximization \cite{CVZ14-crs}.
For a polyhedral relaxation $P$ for $\cI$ and a real $b \in [0,1]$, $b P$
refers to the polyhedron $\{ b x\mid x \in P\}$.

\begin{definition}
  \label{defn:crscheme}
  Let $b,c \in [0,1]$. A $(b,c)$-balanced \crs scheme $\pi$
  for a polyhedral relaxation $P$ for $\cI$ is a procedure that for every $bx \in b P$
  and $A\subseteq  N$, returns a \emph{random} set
  $\pi_x(A)\subseteq A\cap \text{support}(\x)$ and satisfies the following properties:
  \begin{mathproperties}
  \item  $\pi_x(A)\in \mathcal{I}$ with probability $1$
    \;\;$\forall A\subseteq N, x\in b P$, and
  \item\label{item:crs} for all $i \in \text{support}(x)$,
    $\Pr[i\in \pi_x(R(x)) \mid i \in R(x)] \geq c$
    \;\;$\forall x \in b P$.
  \end{mathproperties}
  The scheme is said to be {\em monotone} if $\Pr[i \in
  \pi_x(A_1)]\geq \Pr[i\in \pi_x(A_2)]$ whenever $i \in A_1
  \subseteq A_2$.  A $(1,c)$-balanced \crs scheme is also called a
  \emph{$c$-balanced \crs scheme}. The scheme is \emph{deterministic} if
  $\pi$ is a deterministic algorithm (hence $\pi_x(A)$ is a single set
  instead of a distribution). It is \emph{oblivious} if $\pi$ is
  deterministic and $\pi_x(A) = \pi_y(A)$ for all $x,y$ and $A$,
  that is, the output is independent of $x$ and only depends on $A$.
  The scheme is \emph{efficiently implementable} if $\pi$ is
  a polynomial-time algorithm that given $x,A$ outputs $\pi_x(A)$.
\end{definition}

\section{Approximating via Contention Resolution Schemes}
\label{sec:crs}
Let $G=(V,E)$ be an inductively $k$-independent graph and let
$\vertices = \{v_1,v_2,\ldots,v_n\}$ be the corresponding order. Let
$\cI$ denote the set of independent sets of $G$.  We consider the
following simple polyhedral relaxation for $\cI$ where there is a
variable $x_i$ for each vertex $v_i$. For notational simplicity we let
$\afterneighbors{i}$ denote the set
$N(v_i) \cap \{v_{i+1},\ldots,v_n\}$ which is the set of neighbors of
$v_i$ that come after $v_i$ in the ordering.
\begin{align*}
  x_i + \sum_{v_j \in \afterneighbors{i}} x_j & \le  k \quad \text{for all } i \in [n]\\
  x_i & \in [0,1] \quad \text{for all }i \in [n]
\end{align*}
This is a valid polyhedral relaxation for $\cI$.  Indeed, consider an
independent set $S \subseteq V$, and let $x$ be the indicator vector
of $S$. Fix a vertex $v_i$ and consider the first inequality. If
$v_i \in S$, then since $A_i \subseteq N(v_i)$, we have
$A_i \cap S = \emptyset$, and the left hand side (LHS) is
$1$. Otherwise $\sum_{v_j \in A_i} x_j = \sizeof{A_i \cap S} \leq \alpha(A_i) \leq k$, so the
LHS is at most $k$.

In fact, the $\frac1k$-approximation for MWIS in
\cite{AkcogluADK02,YeB12} are implicitly based on this relaxation.
Moreover, the relaxation has a polynomial number of constraints and
hence is solvable. We refer to this relaxation as $Q_G$ and omit $G$
when clear from the context. The multilinear relaxation is to solve $\max_{x \in Q_G} F(x)$.

Now we consider the case when $G=(V,E)$ is a $k$-perfectly orientable graph.
Let $H=(V,A)$ be an orientation of $G$. For a given $v \in V$
we let $N^+_H(v) = \{u \in V \mid (v,u) \in A\}$ denote the out-neighbors
of $v$ in $H$. We can write a simple polyhedral relaxation for
independent sets in $G$ where we have a variable $x_v$ for each $v \in V$
as follows:
\begin{align*}
  x_v + \sum_{u \in N^+_H(v)} x_u & \le k \quad \text{for all } v \in
                                    V
  \\
  x_v & \in  [0,1] \quad \text{for all } v \in V
\end{align*}

To avoid notational overhead we will use $Q_G$ to refer to the preceding relaxation
for a $k$-perfectly orientable graph $G$. In \cite{KammerT14} a stronger relaxation than
the preceding relaxation is used to obtain a $\frac{1}{2k}$-approximation for
MWIS. It is not hard to see, however, that the proof in \cite{KammerT14}
can be applied to the simpler relaxation above.

We will only consider $k$-perfectly orientable graphs in the rest of this section since the
\crs scheme applies for this more general
class and we do not have a better scheme for inductively $k$-independent graphs.
We consider two simple \crs schemes for $Q$. The first is an oblivious
deterministic one. Given a set $R$ it outputs $S$ where $S = \{ v
\in R \mid  N^+_H(v) \cap R = \emptyset \}$.  In other words it discards
from $R$ any vertex $v$ which has an out-neighbor in $R$.
We claim that $S$ is an independent set. To see this suppose $uv \in E(G)$.
In $H$, $uv$ is oriented as $(u,v)$ or $(v,u)$. Thus, both $u$ and $v$ cannot be
in $S$ even if they are both are picked in $R$. It is also easy to see that the scheme
is monotone.

We now describe a randomized non-oblivious scheme which yields
slightly better constants and is essentially the same as the
one from \cite{Feldman-thesis} where interval graphs were considered
(a special case of $k=1$). This scheme works as follows.
Given $R$ and $x$ it creates a subsample $R' \subseteq R$ by
sampling each $v \in R$ independently with probability $(1-e^{-x_v})/x_v$
(Note that $1-e^{-y} \le y$ for all $y \in [0,1]$.).
Equivalently $R'$ is obtained from $x$ by sampling each $v$
with probability $1-e^{-x_v}$.
It then applies the preceding deterministic scheme to $R'$.
Note that this scheme is randomized and non-oblivious since
it uses $x$ in the sub-sampling step. It is also easy to see
that it is monotone.

We analyze the two schemes.
\begin{theorem}
  For each $b \in [0,1]$ there is a deterministic, oblivious, monotone
  $(b/k,1-b)$ \crs scheme for $Q$. There is a
  randomized monotone $(b/k, e^{-b})$ \crs scheme for $Q$.
\end{theorem}
\begin{proof}
  Let $x \in \frac{b}{k} Q$ and Let $R$ be a random set obtained by picking
  each $v \in V$ independently with probability $x_v$.
  We first analyze the deterministic \crs scheme.
  Fix a vertex $v \in \text{support}(x)$ and condition on
  $v \in R$. The vertex $v$ is included in the final output iff
  $N^+_H(v) \cap R = \emptyset$.
  Since $x \in \frac{b}{k}Q$ we have $\sum_{u \in N^+(v)} x_u \le b - x_v \le b$.
  \begin{align*}
    \Pr[v \in S \mid v \in R]   %
    & =  \Pr[N^+(v) \cap R = \emptyset] 
      = \prod_{u \in N^+(v)} (1-x_u) 
      \ge 1 - \sum_{u \in N^+(v)} x_u 
      \ge 1-b.
  \end{align*}
  This shows that the scheme is a $(b/k,1-b)$ \crs scheme.

  Now we analyze the randomized scheme which follows
  \cite{Feldman-thesis}.  Consider $v \in R(x)$. We see that $v \in S$
  conditioned on $v \in R$, if $v \in R'$ and
  $R' \cap N^+(v) = \emptyset$. Since the vertices are picked
  independently,
  \begin{align*}
    \Pr[v \in S \mid v \in R]   %
    & = \Pr[v \in R' \mid v \in R] \cdot \Pr[N^+(v) \cap R' = \emptyset] 
      = \frac{(1-e^{-x_v})}{x_v} \prod_{u \in N^+(v)} e^{-x_u} \\
    &
      \ge  \frac{(1-e^{-x_v})}{x_v} e^{-(b-x_v)} 
      \ge  \frac{(e^{x_v}-1)}{x_v} e^{-b} %
\ge e^{-b}.
  \end{align*}
  This finishes the proof.
\end{proof}

\medskip
One can apply the preceding \crs schemes for $Q_G$ along with the
known framework via the multilinear relaxation to approximate $\max_{S
  \in \cI} f(S)$. Let $\opt$ be the value of an optimum solution.  For
monotone functions the Continuous Greedy algorithm \cite{CCPV11} can be
used to find a point $x \in \frac{b}{k}Q$ such that $F(x) \ge
(1-e^{-b/k}) \opt$. When combined with the $(b/k,1-b)$ \crs scheme
this yields a $(1-e^{-b/k})(1-b)$-approximation.  The randomized \crs
scheme yields a $(1-e^{-b/k})e^{-b}$-approximation; this bound is
maximized when $b= k \ln (1+1/k)$ and the ratio is $\frac{1}{k+1}
\cdot \frac{1}{(1+1/k)^k} \ge \frac{1}{e (k+1)}$.  For non-negative
functions one can use Measured Continuous Greedy \cite{FeldmanNS11,Feldman-thesis} to obtain $x \in
\frac{b}{k}Q$ such that $F(x) \ge \frac{b}{k} e^{-b/k} \opt$.
Combined with the \crs scheme this yields a  $(\frac{b}{k} e^{-b(1+/k)})$-approximation.
Setting $b = k/(k+1)$ yields a $\frac{1}{e (k+1)}$-approximation.

\begin{theorem}
  There is a randomized algorithm that given a $k$-perfectly orientable graph $G$
  (along with its orientation) and a monotone submodular function $f$, outputs
  an independent set $S'$ such that with high probability
  $f(S') \ge (\frac{1}{k+1}
  \cdot \frac{1}{(1+1/k)^k}) \max_{A \in \cI} f(A)$. For non-negative functions
  there is an algorithm that outputs an independent set $S'$ such that
  with high probability $f(S') \ge \frac{1}{e(k+1)} \max_{A \in \cI} f(A)$.
\end{theorem}

\paragraph{Efficiency and Parallelism.} Approximately solving the multilinear relaxation is
typically a bottleneck. \cite{ChekuriJV15} develops faster algorithms via
the multiplicative-weight update (MWU) based method. We
refer the reader to \cite{ChekuriJV15} for concrete running times that
one can obtain in terms of the number of oracle calls to $f$ or $F$. Once the
relaxation is solved, rounding via the \crs scheme above is simple and
efficient. Another aspect is the design of parallel algorithms, or
algorithms with low adaptivity --- we refer the reader to
\cite{BalkanskiS18} for the motivation and set up.
Via results in \cite{ChekuriQ19,EneNV19}, and the \crs scheme above,
we can obtain algorithms with adaptivity $O(\frac{\log^2 n}{\eps^2})$
while only losing a $(1-\eps)$-factor in the approximation compared
to the sequential approximation ratios.

\section{Primal-Dual Approach for Inductively $k$-Independent Graphs}
\label{sec:primal-dual}
We now consider a primal-dual algorithm. This is inspired by previous
algorithms for MWIS in inductively $k$-independent graphs, and
the work of Levin and Wajc \cite{LevinW21} who considered a primal-dual
based semi-streaming algorithm for submodular function maximization
under matching constraints.

The stack based algorithm in \cite{YeB12} for MWIS is essentially a
primal-dual algorithm. It is instructive to explicitly consider the LP
relaxation and the analysis for MWIS before seeing the algorithm and
analysis for the submodular setting. An interested reader can find
this exposition in \cref{sec:primal-dual-mwis}, which we provide for
the sake of pedagogy and completeness.

Following \cite{LevinW21} we consider an LP relaxation based on the
concave closure of $f$. For independent sets in an inductively
$k$-independent graph, we consider the relaxation
$\max_{x \in Q_G} f^+(x)$. We write this as an explicit LP and
describe its dual. See Fig~\ref{fig:primal-dual-submod}.  The primal has a variable $x_i$ for each
$v_i \in \vertices$ as we saw in the relaxation for MWIS. In
addition to these variables, we have variables
$\alpha_L, L \subseteq \vertices$ to model the objective $f^+(x)$. The
dual has three types of variables. $\mu$ is for the equality
constraint $\sum_L \alpha_L = 1$, $y_i$ is corresponds to the primal
packing constraint for $x_i$ coming from the independence constraint,
and $z_i$ is for the equality constraint coming from modeling $f^+(x)$.

\begin{figure}[thb]
  \centering
  \begin{minipage}{0.45\textwidth}
    \begin{eqnarray*}
      \max \sum_{S \subseteq V} \alpha_L f(L) &&\\
      \sum_{L \subseteq V} \alpha_L & = & 1 \\
      \sum_{L \ni v_i} \alpha_L & = & x_i \quad i \in [n]\\
      x_i + \sum_{v_j \in \afterneighbors{i}} x_j & \le & k \quad i \in [n]\\
      x_i & \ge & 0 \quad i \in [n]
    \end{eqnarray*}
  \end{minipage}
  \vrule
  \begin{minipage}{0.45\textwidth}
    \begin{eqnarray*}
      \min \mu + k \sum_{i=1}^n y_i && \\
      \mu + \sum_{v_i \in L} z_i & \ge & f(L) \quad L \subseteq V\\
      y_i + \sum_{v_j \in \beforeneighbors{i}} y_j & \ge & z_i \quad i \in [n]\\
      y_i & \ge & 0 \quad i \in [n]
    \end{eqnarray*}
  \end{minipage}
  \medskip
  \caption{Primal and Dual LPs via the concave closure relaxation for
    an inductively  $k$-independent graph $\defgraph$ with a given ordering
    $\{v_1,v_2,\ldots,v_n\}$.}
  \label{fig:primal-dual-submod}
\end{figure}

\subsection{Algorithm for monotone submodular functions}
We describe a deterministic primal-dual algorithm for the monotone
case. The algorithm and analysis are inspired by \cite{LevinW21} and we
note that the algorithm has some similarities to the preemptive greedy
algorithm. The primal-dual algorithm takes a two phase
approach similar to algorithm for the modular case.
In the first phase it processes the vertices in the given order and
creates a set $S \subseteq \vertices$. In the second phase it process
the vertices in the reverse order of insertion and creates a maximal
independent set.  Unlike the modular case, the decision to add a
vertex $v_i$ to $S$ in the first phase is based on an inflation factor
$(1+\beta)$. The formal
algorithm is described in Fig~\ref{fig:pd-monotone}. The algorithm creates
a feasible dual as it goes along --- the variables $y,z,\mu$
are from the dual LP. It also maintains and uses auxiliary weight variables
$w_i$, $1 \le i\le n$ that will be useful in the analysis.



\begin{figure}[htb]
  \begin{Inset}
    \begin{algorithm}{primal-dual-monotone-submod}{
        $f:\subsetsof{\vertices} \to \nnreals$,%
        $k \in \naturalnumbers$,%
        $\beta\in\preals$}
    \item Initialize an empty stack $S$. Let
      $\vertices = \setof{v_1,\dots,v_n}$ be a $k$-independence
      ordering of $\vertices$. Set $w,z,y \gets \zeroes_n$.
    \item For $i = 1,\dots,n$:
      \begin{steps}
      \item Let
        \begin{math}
          C_i %
          = %
          N(v_i) \cap S %
          = %
          \setof{u \in S \suchthat \edge{u}{v_i} \in \edges}
        \end{math}
      \item If ($f_S(v_i) > (1 + \beta) \sum_{v_j \in C_i} w_j$) then
        \begin{steps}
        \item Call $S\algo{.push($v_i$)}$ and set $x_i \gets 1$.
        \item Set $w_i \gets f_S(v_i) - \sum_{v_j \in C_i} w_j$ and
          $y_i \gets (1+\beta) w_i$.
        \end{steps}
      \item Otherwise set $z_i \gets f_S(v_i)$
      \end{steps}
    \item Let $\mu \gets f(S)$ and $\Sout \gets \emptyset$
    \item While $S$ is not empty:
      \begin{steps}
      \item $v \gets S\algo{.pop()}$
      \item If $\Sout + v_i$ is independent in $\graph$ then set
        $\Sout \gets \Sout + v_i$.
      \end{steps}
    \item Return $\Sout$
    \end{algorithm}
  \end{Inset}
  \caption{Primal-dual algorithm for monotone submodular
    maximization. The algorithm creates a feasible dual solution in
    the first phase along with a set $\Stackend$. In the second phase it
    processes $\Stackend$ in reverse order of insertion and creates a
    maximal independent set.}
  \label{fig:pd-monotone}
\end{figure}

Let $\Stackend$ be the set of vertices in the stack $S$ at the end of the first
phase. $S$ is a monotonically increasing set during the algorithm.
Note that $\mu = f(\Stackend)$ at the end of the algorithm.
We observe that for each $i$, the algorithm sets the variables $w_i,y_i,z_i$ exactly
once when $v_i$ is processed, and does not alter the values after they
are set.

\begin{lemma}
  The algorithm \algo{primal-dual-monotone-submod} creates a feasible dual solution $\mu,
  \bar{y},\bar{z}$ when $f$ is monotone.
\end{lemma}
\begin{proof}
  We observe that $z_i = 0$ if $v_i \in \Stackend$ and $z_i =
  \incvalue{f}{\setbefore{v_i}}{v_i}$ otherwise.  By submodularity it
  follows that if $v_i \not \in \Stackend$, $z_i \ge
  f_{\Stackend}(v_i)$ since $\setbefore{v_i} \subseteq \Stackend$.

  Consider the first set of constraints in the dual of the form
  $\mu + \sum_{v_i \in L} z_i \ge f(L)$ for $L
  \subseteq V$. We have
  $$\mu + \sum_{v_i \in L} z_i \ge f(\Stackend) + \sum_{v_i \in L \setminus
    \Stackend} f_{\Stackend}(v_i) \ge f(\Stackend \cup L) \ge f(L).$$
  We used submodularity in the second inequality and monotonicity of $f$ in the last inequality.

  Now consider the second set of constraints in the dual of the form
  $y_i + \sum_{v_j \in \beforeneighbors{i}} y_j \ge z_i$ for each $i$. If $v_i
  \in \Stackend$ then $z_i = 0$ and the constraint is trivially
  satisfied since the $y$ variables are non-negative. Assume $v_i \not
  \in \Stackend$. The algorithm did not add $v_i$ to $S$ because
  $$z_i = \incvalue{f}{\setbefore{v_i}}{v_i} \le
  (1+\beta)\sum_{v_j \in C_i} w_j = \sum_{v_j
    \in C_i} y_j$$
  which implies that the constraint for $v_i$ is satisfied.
\end{proof}

Feasibility of the dual solution implies an upper bound on the optimal
value.
\begin{corollary}
  \labelcorollary{opt-bound}
  \begin{math}                  %
    \opt \le f(\Stackend) + k (1+\beta)\sum_{i=1}^n w_i. %
  \end{math}
\end{corollary}

\begin{remark}
  An alternative proof of
  \refcorollary{opt-bound}, not directly based on LP duality or the
  concave closure, is as follows. Let $T$ be a fixed optimal solution.
  We have
  \begin{align*}
    \f{T} - \f{\Stackend}
    &\tago{\leq}
      \sum_{v_i \in T \setminus \Stackend} \f{v_i}_{\Stackend}
      \tago{\leq}
      \sum_{v_i \in T \setminus \Stackend} \incvalue{f}{\setbefore{v_i}}{v_i}
    \\
    &\tago{\leq} \parof{1 + \beta} \sum_{v_i \in T \setminus \Stackend} \sum_{j < i \andcomma v_i \in A_j}
      w_j                       %
      \tago{\leq} \parof{1 + \beta} k \sum_{v_j \in \Stackend} w_j
      = \sum_{i=1}^n w_i.
  \end{align*}
  as desired up to rearrangement of terms.  Here (\tagr*,\tagr*) are
  by submodularity. \tagr is by the inequality in the algorithm that
  excludes each $v_i$ from $\Stackend$. \tagr is because for each
  $v_j \in \Stackend$, $\sizeof{A_j \cap T} \leq k$.
\end{remark}

We now lower bound the value of $f(\Sout)$.\needspace{2em}
\begin{lemma}
  \labellemma{lower-bound-pd} $f(\Sout) \ge \sum_{i=1}^n w_i$.
\end{lemma}
\begin{proof}
  A vertex $v_i$ is added to $\Stackend$ since
  $\incvalue{f}{\setbefore{v_i}}{v_i} > (1+\beta) \sum_{v_j \in C_i}
  w_j$. Moreover, we have $w_i + \sum_{v_j \in C_i} w_j =
  \incvalue{f}{\setbefore{v_i}}{v_i}$ via the algorithm.
  Therefore,
  $$f(\Sout) = \sum_{v_i \in \Sout} \incvalue{f}{\Sout}{v_i} \ge
  \sum_{v_i \in \Sout}\incvalue{f}{\setbefore{v_i}}{v_i} = \sum_{v_i
    \in \Sout} (w_i + \sum_{j \in C_i} w_j).$$
  We see that for every $i'$ such that $v_{i'} \in \Stackend$ the term
  $w_{i'}$ appears at least once in $\sum_{v_i
    \in \Sout} (w_i + \sum_{j \in C_i} w_j)$; either $v_{i'} \in
  \Sout$ or if it is not then it was removed in the second
  phase since $v_{i'} \in C_i$ for some $v_i \in \Sout$. In the
  latter case $w_{i'}$ appears in the $\sum_{j \in C_i} w_j$.
  Thus $f(\Sout) \ge \sum_{i=1}^n w_i$ (recall that $w_i = 0$ if
  $v_i \not \in \Stackend$).
\end{proof}

We now upper bound $f(\Stackend)$ via the weights.
\begin{lemma}
  \labellemma{upper-bound-mu}
  \begin{math}
  f(\Stackend) \le \frac{1+\beta}{\beta} \sum_{i=1}^n w_i.
\end{math}
\end{lemma}
\begin{proof}
  Let $v_i \in \Stackend$. Recall that
  $\incvalue{f}{\setbefore{v_i}}{v_i} \geq (1+\beta)\sum_{j \in C_i} w_j$
  and
  $w_i = \incvalue{f}{\setbefore{v_i}}{v_i} - \sum_{j \in C_i} w_j$.
  This implies that
  \begin{math}
    w_i \ge \frac{\beta}{1+\beta} \incvalue{f}{\setbefore{v_i}}{v_i}.
  \end{math}
  Therefore
  \begin{align*}
    f(\Stackend) = \sum_{v_i\in \Stackend} \incvalue{f}{\Stackend}{v_i} =
    \sum_{v_i \in \Stackend} \incvalue{f}{\setbefore{v_i}}{v_i} \le
    \frac{1+\beta}{\beta} \sum_{v_i \in \Stackend} w_i,
  \end{align*}
  as desired.
\end{proof}


\begin{theorem}
  \begin{math}
    \opt \le (1+\beta)\parof{1/\beta + k} f(\Sout).
  \end{math}
  In particular, for $\beta = \frac{1}{\sqrt{k}}$,
  \begin{math}
    \opt \le (k+1+2\sqrt{k}) f(\Sout).
  \end{math}
\end{theorem}
\begin{proof}
  From \refcorollary{opt-bound} and  \reflemma{upper-bound-mu}
  and \reflemma{lower-bound-pd},
  \begin{align*}
    \opt & \le f(\Stackend) + k (1+\beta)\sum_{i=1}^n w_i %
           \le
           \frac{1+\beta}{\beta}\sum_{i=1}^n w_i +
           k(1+\beta)\sum_{i=1}^n w_i\\
         & \le (1+\beta)(\frac{1}{\beta} + k) \sum_{i=1}^n w_i %
           \le(1+\beta)(\frac{1}{\beta} + k) f(\Sout),
  \end{align*}
  as desired.
\end{proof}

\begin{remark}
  For $k=1$ we obtain a $1/4$-approximation which yields a
  deterministic $1/4$-approximation for chordal graphs and
  interval graphs. For $k=2$ we obtain a bound of $3+2\sqrt{2}$ which
  is the same as what \cite{LevinW21} obtain for matchings. Note that
  matchings can be interpreted, via the line graph, as inductive
  $2$-independent and in fact any ordering of the edges is an
  inductive $2$-independent order. This explains why the ordering does
  not matter. \cite{LevinW21} use a different LP relaxation
  for matchings, and hence it is a bit surprising that we obtain the same bound
  for all $2$-independent graphs. For the non-monotone case
  we obtain a weaker bound for $2$-independent graphs than what \cite{LevinW21} obtain for matchings.
\end{remark}

\subsection{Non-monotone submodular maximization}
We now consider the case of non-negative submodular function which may not
be necessarily monotone. This class of functions
requires some additional technical care and a key lemma that is useful
in handling non-monotone function is the following.

\begin{lemma}[\cite{bfns-smcc-14}]
  \labellemma{bfns-nonneg}
  Let $f:2^V \rightarrow \mathbb{R}_+$ be a non-negative submodular function.
  Fix a set $T \subseteq V$. Let $S$ be a random subset of $V$ such that
  for any $v \in V$ the probability of $v \in S$ is at most $p$ for some $p < 1$.
  Then $\evof{f(S \cup T)} \ge (1-p) f(T)$.
\end{lemma}

We describe a randomized primal-dual algorithm which is adapted from the
one form \cite{LevinW21}. It differs from the monotone algorithm in one
simple but crucial way; even when a vertex $v$ has good value compared
to its conflict set it adds it to the stack only with probability $p$ which
is a parameter that is chosen later.

\begin{figure}[htb]
  \begin{Inset}
    \begin{algorithm}{primal-dual-nonneg-submod}{%
        $f:\subsetsof{\vertices} \to \nnreals$,%
        $k \in \naturalnumbers$,%
        $\beta\in\preals$}
    \item
      Initialize an empty stack $S$. Let $\vertices = \setof{v_1,\dots,v_n}$ be a
      $k$-independence ordering of $\vertices$. Let $w, y, z = \zeroes_n$.
    \item For $i = 1,\dots,n$:
      \begin{steps}
      \item Let
        \begin{math}
          C_i %
          = %
          N(v_i) \cap S %
          = %
          \setof{u \in S \suchthat \edge{u}{v_i} \in \edges}
        \end{math}
      \item If ($f_S(v_i) > (1 + \beta) \sum_{v_j \in C_i} w_j$),
        then with probability $p$:
        \begin{steps}
        \item Call $S\algo{.push($v_i$)}$ and $x_i \leftarrow 1$
        \item Set $w_i \gets f_S(v_i) - \sum_{v_j \in C_i} w_j$ and
          $y_i \gets (1+\beta) w_i$
        \end{steps}
      \item Otherwise set $z_i \gets f_S(v_i)$.
      \end{steps}
    \item Set $\mu \gets f(S)$ and $\Sout \gets \emptyset$
    \item While $S$ is not empty:
      \begin{steps}
      \item $v \gets S\algo{.pop()}$.
      \item If $\Sout + v_i$ is independent in $\graph$ then set
        $\Sout \gets \Sout + v_i$
      \end{steps}
    \item Return $\Sout$
    \end{algorithm}
  \end{Inset}
  \caption{Randomized primal-dual algorithm for non-negative submodular
    maximization.}
\end{figure}

\paragraph{Analysis:} As in the monotone case let $\Stackend$ be the set of vertices in the stack
at the end of the first phase (note that $\Stackend$ is now a random set).
The analysis of the randomized version of the algorithm is technically
more involved.  The sets $\Stackend, \Sout$ and the dual variables are now
random variables. Since very high-value vertices can be discarded
probabilistically, the dual values constructed by the algorithm may
not satisfy the dual constraints for each run of the algorithm.  Levin
and Wajc \cite{LevinW21} analyze their algorithm for matchings via an ``expected''
dual solution.  We do a more direct analysis via weak duality.

The following two lemmas are essentially the same as in the monotone
case and they relate the expected value of $\Sout$ and $\Stackend$ to
the dual weight values.

\begin{lemma}
  \labellemma{lower-bound-pd-nonneg}
  For each run of the algorithm: $f(\Sout) \ge \sum_{i=1}^n w_i$ and
  hence
  \begin{math}
  \evof{f(\Sout)} \ge \sum_{i=1}^n \evof{w_i}.
\end{math}
\end{lemma}

\begin{lemma}
  \labellemma{upper-bound-mu-nonneg}
  For each run of the algorithm,
  $f(\Stackend) \le \frac{1+\beta}{\beta} \sum_{i=1}^n w_i$ and hence
  $$\evof{f(\Stackend)} \le \frac{1+\beta}{\beta} \sum_{i=1}^n\evof{w_i}.$$
\end{lemma}

The next two lemmas provide a way to upper bound the optimum value via the
expected dual objective value.

\begin{lemma}
  \labellemma{expected-dual-contrib}
  For each vertex $v_i$, let $1_{v_i \notin \Stackend}$ indicate if $v_i$ is
  excluded from $\Stackend$. Let $\beforeneighbors{i}' = \beforeneighbors{i}+v_i$.   Then
  \begin{align*}
    \evof{\f{v_i}{S_i}1_{v_i \notin \Stackend}}
    \leq                        %
    \max{\frac{1-p}{p}, 1 + \beta} %
    \evof{w(\beforeneighbors{i}' \cap \Stackend)}.
  \end{align*}
\end{lemma}

\begin{proof}
  Let $E_i$ be the event that
  \begin{math}
    \f{v_i}{\Stack_i} > \parof{1 + \beta} w(\beforeneighbors{i}
    \cap \Stack_i).
  \end{math}
  Condition on $\bar{E_i}$, that is $E_i$ \emph{not} occurring, in which case $v_i$ is not added to the stack.
  In this case we have
  \begin{align*}
    \evof{w(\beforeneighbors{i})}{\bar{E}_i}
    & =
      \evof{w(\beforeneighbors{i}' \cap \Stackend)}{\bar{E}_i}
      \geq
      \frac{1}{1 + \beta} \evof{\f{v_i}{\Stack_i}}{\bar{E}_i}
    \\
    &=
      \frac{1}{1 + \beta} \evof{\f{v_i}{\Stack_i} 1_{v_i \notin \Stackend}}{\bar{E}_i}
  \end{align*}
  On the other hand, condition on $E_i$, we have
  \begin{align*}
    \evof{w(\beforeneighbors{i}' \cap \Stackend)}{E_i} %
    \tago{\geq}                                       %
    p \evof{\f{v_i}{S_i}}{E_i}
    \tago{=}
    \frac{p}{1-p} \evof{\f{v_i}{S_i} 1_{v_i \notin \Stackend }}{E_i}.
  \end{align*}
  \tagr is because with probability $p$, we add $v$ to the stack, in
  which case $w(\beforeneighbors{i}' \cap \Stackend) \geq \f{v_i}{S_i}$.
  \tagr is because conditional on $E_i$ and $\f{v_i}{S_i}$,
  $v_i \notin \Stackend$ with probability $1-p$.  We combine the two
  bounds by taking conditional expectations, as follows:
  \begin{align*}
    \evof{\f{v_i}{S_i}1_{v_i \notin \Stackend}}
    &=                       %
      \evof{\f{v_i}{S_i}1_{v_i \notin \Stackend}}{E_i}\probof{E_i} +
      \evof{\f{v_i}{S_i}1_{v_i \notin \Stackend}}{\bar{E}_i} \probof{\bar{E_i}} \\
    &\leq                       %
      \frac{1-p}{p} \evof{w(\beforeneighbors{i}' \cap \Stackend)}{E_i}
      \probof{E_i} %
      +                         %
      \parof{1 + \beta}  \evof{w(\beforeneighbors{i}' \cap \Stackend)}{\bar{E}_i}
      \probof{\bar{E}_i}
    \\
    &\leq                       %
      \max{\frac{1-p}{p}, 1 + \beta}
      \parof{ \evof{w(\beforeneighbors{i}' \cap \Stackend)}
      \probof{E_i} %
      +    \evof{w(\beforeneighbors{i}' \cap \Stackend)}                     %
      \probof{\bar{E}_i} %
      }                                                      %
    \\
    &=
      \max{\frac{1-p}{p}, 1 + \beta} \evof{w(\beforeneighbors{i}' \cap \Stackend)},
  \end{align*}
  as desired.
\end{proof}

\begin{lemma}
  \labellemma{expected-dual-val}
  For any set $T$, $\evof{\f{\Stackend \cup T}} \leq \evof{\f{\Stack}} + k
    \max{\frac{1-p}{p}, 1 + \beta} \evof{w(\Stackend)}.$
\end{lemma}
\begin{proof}
  We have
  \begin{align*}
    \evof{\f{T \cup \Stackend} - \f{\Stackend}}
    &\tago{\leq}                        %
      \evof{\sum_{v_i \in T \setminus \Stackend} \f{v_i}{\Stackend}} %
      \tago{\leq}
      \evof{\sum_{v_i \in T \setminus \Stackend} \f{v_i}{\Stack_i}}
    \\
    &=                          %
      \sum_{v_i \in T} \evof{\f{v_i}{\Stack_i} 1_{v_i \notin \Stackend}}
      \tago{\leq}
      \max{\frac{1-p}{p}, 1 + \beta} \sum_{v_i \in T}
      \evof{w(\beforeneighbors{i}' \cap \Stackend)}    \\
    &\tago{\leq}
      \max{\frac{1-p}{p}, 1 + \beta} k \evof{w(\Stackend)},
  \end{align*}
  as desired up to rearrangement of terms.  Here (\tagr*,\tagr*) is by
  submodularity. \tagr is by the  \reflemma{expected-dual-contrib}. \tagr is by
  $k$-inductive independence.
\end{proof}

We now put the lemmas together to relate $\evof{f(\Sout)}$ to the optimum.

\begin{lemma}
  \label{lem:pd-nonnon-summary}
  Let $T^*$ be an optimum independent set with $\opt =
  f(T^*)$. Then
  \begin{align*}
    \opt &\leq                   %
           \frac{k\max{\frac{1-p}{p}, 1 + \beta} +
           \prac{1+\beta}{\beta}}{1-p} \evof{\f{\Sout}}.
  \end{align*}
\end{lemma}

\begin{proof}
  Let $T$ be any independent set, in particular $T^*$. We observe that the algorithm ensures that
  for any vertex $v$, $\probof{v \in \Stackend} \le p$ and hence $\probof{v \in \Sout} \le p$.
  \begin{align*}
    \parof{1-p}\f{T}
    &\leq
      \evof{\f{T \cup \Stackend}} \quad \text{(\reflemma{bfns-nonneg})}\\
    &\leq \evof{\f{\Stack}} + k\max{\frac{1-p}{p}, 1 + \beta}
      \evof{w(\Stackend)} \quad \text{(\reflemma{expected-dual-val})}\\
    &\leq                         %
      \parof{k\max{\frac{1-p}{p}, 1 + \beta} +
      \prac{1+\beta}{\beta}} \evof{w(\Stackend)} \quad \text{(\reflemma{upper-bound-mu-nonneg})} \\
    &\leq
      \parof{k\max{\frac{1-p}{p}, 1 + \beta} +
      \prac{1+\beta}{\beta}} \evof{\f{\Sout}} \quad \text{(\reflemma{lower-bound-pd-nonneg})}.
  \end{align*}
\end{proof}

It remains to choose $p \in [0,1]$ and $\beta > 0$ to minimize the
RHS. Consider the term $\max{(1-p)/p, 1+\beta}$. If
$(1-p)/p \geq 1 + \beta$, then $p \geq 1/2$ (to force
$(1-p) / p \geq 1$), and the RHS is minimized by taking $\beta$ as
large as possible -- that is, such that $1 + \beta = (1-p)/p$. If
$(1-p)/p \leq 1 + \beta$, then the RHS is minimized by taking $p$ as
small as possible -- that is, such that $(1-p)/p = 1/p-1 =
1+\beta$. Thus $(1-p)/p = 1 + \beta$ at the optimum. In terms of just
$p$, then, we have
\begin{align*}
  \opt &\leq
         \prac{1}{1-p} \parof{\frac{k(1-p)}{p} + \frac{1-p}{1-2p}}
         \evof{\f{\Sout}}
         =                             %
         \parof{\frac{k}{p} + \frac{1}{1-2p}} \evof{\f{\Sout}}.
\end{align*}
(Here we note that $\beta = (1-2p) / p$, hence
$(1+\beta) / \beta = (1 - p) / (1-2p)$.)

In the special case of $k = 2$, as in matching, the RHS is
\begin{align*}
  \opt \leq
  \parof{\frac{2}{p} + \frac{1}{1 - 2p}} \evof{\f{\Sout}}
\end{align*}
The RHS is minimized
by $p = 1/3$, giving an approximation
factor of $9$.

For general $k$, the minimum is 
\begin{math}
  2 k + \sqrt{8k} + 1.
\end{math}

It is easy to see that the primal-dual algorithm makes $O(n)$
evaluation calls to $f$ and the overall running time is linear in the
size of the graph. The results for the monotone and non-negative
functions, together yield Theorem~\ref{thm:intro-primal-dual}.

\section{A Preemptive Greedy Algorithm}
\label{sec:greedy}
We now describe a preemptive greedy algorithm for maximizing a
monotone submodular function $f: 2^V \rightarrow \mathbb{R}_+$ over
independent sets of a inductively $k$-independent graph $\defgraph$
assuming that we are also given the ordering.  The algorithm is simple
and intuitive, and is inspired by algorithms developed in the streaming model.

The pseudocode for the algorithm is given in 
\reffigure{preemptive-greedy}, and is designed as follows.  Starting from
an empty solution $S = \emptyset$, \algo{preemptive-greedy} processes
the vertices in the given ordering one by one. When considering $v_i$,
the algorithm gathers the subset $C_i \subseteq S$ of all vertices in
the current set $S$ that are neighbors of $v_i$ (those that conflict
with $v_i$). The algorithm has to decide whether to reject $v_i$ or to
accept $v_i$ in which case it has to remove $C_i$ from $S$.  It
accepts $v_i$ if the marginal gain $f_S(v_i) \defeq f(S + v_i) - f(S)$
of adding $v_i$ directly to $S$ is at least $(1+\beta)$ times the
value $\sum_{u \in C_i}f_{S \setminus C_i}(u)$.  Here $\beta > 0$ is a
parameter that is fixed based on the analysis.  After processing all
vertices, we return the final set $S$.

\begin{figure}[htb]
  \begin{Inset}
    \begin{algorithm}{preemptive-greedy}%
      {$\defgraph$,%
        $f:\subsetsof{\vertices} \to \nnreals$,%
        $k \in \naturalnumbers$,%
        $\beta\in\preals$}

    \item Let $S = \emptyset$. Let $\vertices = \setof{v_1,\dots,v_n}$
      by a $k$-independence ordering of $\vertices$
    \item For $i = 1,\dots,n$:
      \begin{steps}
      \item Let
        \begin{math}
          C_i %
          = %
          N(v_i) \cap S %
          = %
          \setof{u \in S \suchthat \edge{u}{v_i} \in \edges}
        \end{math}
      \item If
        \begin{math}
          f_S(v_i) \geq (1 + \beta) \sum_{u \in C_i} \incv{f}{S}{u}
        \end{math}
        \begin{steps}
        \item \labelstep{preemptive-greedy-augment} Set
          \begin{math}
            S \gets (S \setminus C_i) + v_i
          \end{math}
        \end{steps}
      \end{steps}
    \item Return $S$
    \end{algorithm}
  \end{Inset}
  \caption{The algorithm \algo{preemptive-greedy} for
    finding an independent set in a inductively $k$-independent
    graph to maximize a monotone
    submodular objective function. \labelfigure{preemptive-greedy}}
\end{figure}

\algo{preemptive-greedy} for inductively $k$-independent graphs has
the following bounds. As the bounds are slightly weaker than the ones
given by the primal dual algorithms, the analysis is deferred to
\refsection{preemptive-greedy-analysis}.
\begin{theorem}
  Given an inductively $k$-independent graph with a $k$-inductive
  ordering, the algorithm \algo{preemptive-greedy} returns an
  independent set $\Sout$ such that for any independent set $T$,
  $f(T) \leq (k(1 + \beta) + 1)(1 + \beta^{-1}) f(\Sout).$
\end{theorem}

\algo{preemptive-greedy} can be extended to nonnegative (and
non-monotone) submodular functions with a constant factor loss in
approximation by random sampling. As a preprocessing step, we let
$\vertices'$ randomly sample each vertex in $\vertices$ independently with
probability $1/2$. We then apply \algo{preemptive-greedy} to the
subgraph $\graph' = \graph[\vertices']$ induced by $\vertices'$.
It is easy to see that any subgraph of an inductively $k$-independent graph is also
inductively $k$-independent. The net effect of the random sampling is an
approximation factor for nonnegative submodular functions that is a
factor $4$ worse than for the monotone case. The modified algorithm,
called \algo{randomized-preemptive-greedy}, is given in
\reffigure{randomized-preemptive-greedy}.

\begin{figure}
  \begin{Inset}
    \begin{algorithm}{randomized-preemptive-greedy}%
      {$\defgraph$,%
        $f:\subsetsof{\vertices} \to \nnreals$,%
        $k \in \naturalnumbers$,%
        $\beta\in\preals$}
    \item Let $\vertices' \subseteq \vertices$ sample each $v \in
      \vertices$ independently with probability $1/2$
    \item Let $\graph' = \graph[\vertices']$ be the subgraph of $\graph$ induced by $\vertices'$
    \item Return
      \algo{preemptive-greedy(%
        $\graph'$,%
        $f:\subsetsof{\vertices'} \to
        \nnreals$,%
        $k \in \naturalnumbers$,%
        $\beta\in\preals$%
        )}.\
    \end{algorithm}
  \end{Inset}
  \caption{The algorithm \algo{randomized-preemptive-greedy} for
    finding an independent set in an inductively $k$-independent
    graph to maximize a nonnegative
    submodular objective function. \labelfigure{randomized-preemptive-greedy}}
\end{figure}

\begin{theorem}
  Given an inductively $k$-independent graph with a $k$-inductive
  ordering, the algorithm \algo{randomized-preemptive-greedy} returns
  an independent set $\Sout$ such that for any independent set $T$,
  $f(T) \leq 4(k(1 + \beta) + 1)(1 + \beta^{-1}) f(\Sout).$
\end{theorem}

\begin{remark}
  The randomized strategy we outline is simple and oblivious. It
  loses a factor of $4$ over the monotone case. One could try to
  improve the approximation ratio by using randomization within the
  algorithm which would make the analysis more involved. However, we
  have not done this since the primal-dual algorithm yields better
  approximation bounds. This subsampling strategy is not new and has
  been used previously in \cite{FeldmanKK18}, and is also implicit in
  \cite{ChekuriGQ15}.
\end{remark}

\section{Concluding Remarks and Open Problems}
\label{sec:concl}
We described $\Omega(\frac{1}{k})$-approximation algorithms for
independent sets in two parameterized families of graphs that capture
several problems of interest. Although the multilinear relaxation
based framework yields such algorithms, the resulting algorithms are
computationally expensive and randomized.  We utilized ideas from
streaming and primal-dual based algorithms to give simple and fast
algorithms for inductively $k$-independent graphs with the additional
property that they are deterministic for monotone functions.  Our work
raises several interesting questions that we summarize below.

\begin{itemize}
\item The CR scheme that we described in Section~\ref{sec:crs} is
  unable to distinguish $k$-perfectly orientable graphs and inductive
  $k$-independent graphs. Is a better bound possible for inductively
  $k$-independent graphs?
\item Our combinatorial algorithms only apply to inductively
  $k$-independent graphs. Can we obtain combinatorial algorithms for
  $k$-perfectly orientable graphs?  Even for MIS the only approach
  appears to be via primal rounding of the LP solution
  \cite{KammerT14}.
\item Can we obtain deterministic $\Omega(\frac{1}{k})$-approximation
  algorithms for these graph classes when $f$ is non-negative?
  Interval graphs seem to be a natural first step to consider.
\item Are better approximation ratios achievable? For instance, can we
  obtain better than $1/4$-approximation for monotone submodular
  function maximization in interval graphs? Can we prove better lower
  bounds under complexity theory assumptions or in the oracle model
  for interval graphs or other concrete special cases of interest?
\item For both classes of graphs our algorithms are based on having an
  ordering that certifies that they belong to the class.  For MWIS in
  $k$-simplicial and $k$-perfectly orientable graphs,
  \cite{HalldorssonT21} describes algorithms based on the Lov\'asz
  number of a graph and the Lov\'asz $\theta$-function of a graph, and
  these algorithms do not require an ordering.  It may be feasible to
  extend their approach to the submodular setting via the multilinear
  relaxation. However, the resulting algorithms are computationally
  quite expensive. It would be interesting to obtain fast algorithms
  for these classes of graphs (or interesting special cases) when the
  ordering is not explicitly given.
\end{itemize}


\printbibliography

\appendix


\section{Interpreting the $k$-approximation for MWIS via
  primal-dual}\label{sec:primal-dual-mwis}
For the sake of completeness we show that the stack
based algorithm in \cite{YeB12} can be interpreted as a primal-dual
approximation algorithm via the LP relaxation $Q_G$ that we saw
previously. We state the primal and dual LPs below.  The primal is the MWIS
LP with $x_i$ denoting whether $v_i$ is chosen in the independent
set. The dual can be seen as a covering LP.  Recall that
$\afterneighbors{i}$ is the set of neighbors of $v_i$ that come after
it in the ordering. Similarly we let
$\beforeneighbors{i} = N(v_i) \cap \{v_1,\ldots,v_{i-1}\}$ denote the
set of neighbors of $v_i$ that come before $v_i$ in the ordering.

\begin{figure}[htb]
  \centering
  \begin{minipage}{0.45\textwidth}
    \begin{eqnarray*}
      \max \sum_{i=1}^n w_i x_i &&\\
      x_i + \sum_{v_j \in \afterneighbors{i}} x_j & \le & k \quad i \in [n]\\
      x_i & \ge & 0 \quad i \in [n]
    \end{eqnarray*}
  \end{minipage}
  \vrule
  \begin{minipage}{0.45\textwidth}
    \begin{eqnarray*}
      \min k \sum_{i=1}^n y_i && \\
      y_i + \sum_{v_j \in \beforeneighbors{i}} y_j & \ge & w_i \quad i \in [n]\\
      y_i & \ge & 0 \quad i \in [n]
    \end{eqnarray*}
  \end{minipage}

  \medskip
  \caption{Primal and Dual LPs for MWIS in an inductively
    $k$-independent graph $\defgraph$ with a given ordering
    $\{v_1,v_2,\ldots,v_n\}$.}
  \label{fig:primal-dual-mwis}
\end{figure}

\begin{remark}
  We observe that the primal LP relaxation does not enforce the
  condition that $x_i \le 1$. Thus the relaxation allows up to $k$ copies
  of a vertex to be chosen. The primal-dual algorithm chooses at most one copy of a
  vertex. The analysis shows that the integrality gap is at most $1/k$
  even with the relaxation. The advantage of dropping the $x_i \le 1$
  constraints is a simpler dual.
\end{remark}

The primal-dual algorithm is described in
Fig~\ref{fig:primaldual-mwis}. It has two phases, a growing phase in
which a set $S \subseteq \vertices$ is created.  This is guided by a
dual solution $y$ which processes vertices in the inductive
$k$-independent order.  One can see this set $S$ as the stack produced
in the algorithm in \cite{YeB12}. To be consistent we use a stack for
$S$. In the second phase the vertices in $S$ are processed in the
reverse order to create a maximal independent set $\Sout$.


\begin{figure}
  \begin{Inset}
    \begin{algorithm}{primal-dual-mwis}
      {$\defgraph$,
        $k \in \naturalnumbers$}
    \item Initialize and empty stack $S \leftarrow
      \emptyset$. Initialize primal and dual solutions  $x \gets
      \zeroes$ and $y \gets \zeroes$.
    \item for $i = 1,\ldots,n$:
      \begin{steps}
      \item Let $y_i = \max\{0, w_i - \sum_{j < i, v_i \in A_j} y_j\}$
      \item If $y_i > 0$ then call $S\algo{.push($v_i$)}$ and set
        $x_i = 1$.
      \end{steps}
    \item Let $\Sout \gets \emptyset$.
    \item While $S$ is not empty:
      \begin{steps}
      \item $v \gets S\algo{.pop()}$
      \item If $\Sout \cap N(v) = \emptyset$ then set
        $\Sout \gets \Sout + v$
      \end{steps}
    \item Output $\Sout$
    \end{algorithm}
  \end{Inset}
  \caption{Primal-dual for MWIS in inductively $k$-independent graphs.}
  \label{fig:primaldual-mwis}
\end{figure}

We now analyze the algorithm. Let $\Stackend$ be the set of vertices
in the stack $S$ at the end of the first phase, and let $\Sout \subseteq \Stackend$ be
the final output of the algorithm. The following observations are easy
to verify and we omit a formal proof.
\begin{itemize}
\item The algorithm constructs a feasible dual solution $y$.
\item A vertex $v_i \in \Stackend$ (that is $x_i = 1$) during the
  first phase iff $y_i > 0$. Further, if $v_i \in \Stackend$ then the
  dual constraint for $v_i$ is tight:
  $y_i + \sum_{j< i, v_i \in A_j} y_j = w_i$.  In other words the
  algorithm maintains primal complementary slackness condition for
  $x,y$.
\item $\Sout$ is a maximal independent set in $G[\Stackend]$.
\end{itemize}

We note that the dual solution $y$ is the same as the adjusted weights
created by  the stack based algorithm in \cite{YeB12}.
The key claim is the  following:
\begin{lemma}
  \labellemma{weights>duals}
  Let $\Sout$ be the set of vertices output by the algorithm. Then
  $w(\Sout) \ge \sum_{i} y_i$ where $y$ is the dual constructed by the algorithm.
\end{lemma}
\begin{proof}
  We observed that $x_i = 1$ implies that $y_i + \sum_{j < i, v_i \in
    A_j} y_j = w_i$.
  Hence, $w(\Sout) = \sum_{v_i \in \Sout} (y_i + \sum_{j<i, v_i \in
    A_j} y_j)$.
  Suppose $v_{i'} \in \Stackend \setminus \Sout$, then $v_{i'}$ was
  considered in the second phase but was not included since
  there was some $i > i'$ such that $v_i \in \Sout$ and
  $v_i \in A_{i'}$. This implies that $y_{i'}$ is counted in the
  term $\sum_{j < i, v_i \in A_{j}} y_j$. Thus, for every $v_j \in
  \Stackend$, the dual variable $y_j$ is included
  in the sum, and hence,
  $$w(\Sout) = \sum_{v_i \in \Sout} (y_i + \sum_{j<i, v_i \in
    A_j} y_j) \ge \sum_{i \in \Stackend} y_i = \sum_i y_i$$ where the last
  equality follows from the fact that $v_i \in \Stackend$ iff $y_i > 0$.
\end{proof}

\medskip
Thus the algorithm outputs a feasible independent set
$\Sout$ such that $w(\Sout) \ge \sum_i y_i$. But note that dual value
$k \sum_i y_i$ is an upper bound on the optimum LP value $\opt_{LP}$
since $y$ is feasible. Thus $w(\Sout) \ge \frac{1}{k} \opt_{LP}$
which proves that the algorithm yields a $\frac{1}{k}$ approximation
with respect to the LP relaxation $Q_G$  (in fact the relaxation that
drops
the constraints $x_i \le 1, i \in [n]$).

\section{Analysis of preemptive greedy}
\labelsection{preemptive-greedy-analysis}

We follow the notation of
\cite{ChekuriGQ15}.  Let $\Sout$ be the final set of vertices returned by
\algo{preemptive-greedy}. It is easy to see that the algorithm returns
an independent set. For each  $u \in \vertices$
let $\schedulebefore{u}$ denote the set of vertices in $S$ just before
$u$ is processed, and let $\scheduleafter{u}$ denote the set
after $u$ is processed. Thus a vertex $u$ is added to $S$
iff $\scheduleafter{u} \setminus \schedulebefore{u} = \setof{u}$. Let
\begin{math}
  U = \bigcup_{u \in \vertices} \scheduleafter{u}
\end{math}
be the set of all vertices that were ever (even momentarily) added to
$S$. Alternatively, $\vertices \setminus U$ is the set of vertices
that are discarded by the algorithm when it considers them.  For each
vertex $u$, let
$\gain{u} \defeq f(\scheduleafter{u}) - f\parof{\schedulebefore{u}}$
be the value added to $S$ from processing $u$. We have $\gain{u}= 0$
for all $u \notin U$, and
\begin{math}
  f(\Sout) = \sum_{u \in \vertices} \gain{u} = \sum_{u \in U}
  \gain{u}.
\end{math}

Let $T \subseteq \vertices$ be an independent set in the given graph,
in particular an optimum set. We would like to compare $f(\Sout)$ with
$f(T)$. Directly comparing $T$ with $\Sout$ is difficult since $\Sout$
is obtained by deleting vertices in $S$ along the way; thus a vertex
$v \in T \setminus \Sout$ may have been discarded due to a vertex
$u \in S$ when $v$ was considered but $u$ may not be in $\Sout$.
Thus, the analysis is broken into two parts that detour through
$U$. First, we relate the value of $f(\Sout)$ to the value of
$f(U)$. This part of the analysis bounds the amount of value lost by
kicking out vertices from $S$ during the exchanges. We then relate
$f(U)$ and $f(T)$; this is easier because any vertex in $T$ is always
compared against \emph{some} subset of vertices in $U$. Chaining the
inequalities from $f(\Sout)$ to $f(U)$ to $f(T)$ gives the final
approximation ratio.

\paragraph{Relating $f(\Sout)$ to $f(U)$:} The analysis is similar to
that in \cite{ChekuriGQ15}. We provide proofs for the sake of
completeness.  The following claim is easy to see since elements
before $s$ can only be deleted from $S$ as the algorithm proceeds.
\begin{claim}
  Over the course of the algorithm, the incremental value
  $\incv{f}{S}{s}$ of an element $s \in S$ is nondecreasing.
\end{claim}

For a vertex $u \in U \setminus \Sout$ we let $u'$ denote the vertex
that caused $u$ to be removed from $S$. And we let $\exitvalue{u}$ denote its
incremental value just before it is removed. Therefore,
$\chi(u) = \incv{f}{\setbefore{u'}}{u}$.

\begin{lemma}
  \labellemma{deltavalue}
  Let $u \in U$ then $\delta_u \ge \beta \sum_{c \in C_u} \incv{f}{\setbefore{u}}{c}$.
\end{lemma}
\begin{proof}
  Since the vertex $u$ was added to $S$ when it was considered, we
  have $\delta_u = f(\setafter{u}) - f(\setbefore{u})$ where
  $\setafter{u} = \setbefore{u} - C_u + u$.
  The vertex $u$ was added by the algorithm since $f_S(u) \ge (1+\beta) \sum_{c
    \in C_u}\incv{f}{S}{c}$ where $S = \setbefore{u}$. Therefore
  $\beta \sum_{c \in C_u} \incv{f}{\setbefore{u}}{c} \le f_{\setbefore{u}}(u) - \sum_{c \in C_u} \incv{f}{\setbefore{u}}{c}$.
  It suffices to prove that $f(\setafter{u}) -
  f(\setbefore{u}) \ge f_S(u) - \sum_{c \in C_u}\incv{f}{S}{c}$ which we do below.
  For notational convenience let $A = \setbefore{u} - C_u$.
  \begin{align*}
    f(\setafter{u}) -  f(\setbefore{u}) & = f(A + u) - f(\setbefore{u}) \\
                                        & = f_A(u) + f(A) - f(\setbefore{u}) \\
                                        & \ge f_{\setbefore{u}}(u) - (f(\setbefore{u}) - f(A)) & \text{by submodularity since $A \subseteq \setbefore{u}$}\\
                                        & \ge f_{\setbefore{u}}(u) - \sum_{c \in C_u} \incv{f}{\setbefore{u}}{c} & \text{by submodularity and defn of $\nu$.}
  \end{align*}
\end{proof}

\begin{lemma}
  \labellemma{exit-values-upper-bound}
  $\sum_{u \in U\setminus \Sout} \exitvalue{u} \leq  \beta^{-1} f(\Sout)$.
\end{lemma}
\begin{proof}
  Indeed,
  \begin{align*}
    \sum_{u \in U \setminus \Sout} \exitvalue{u}
    & = \sum_{u \in U} \sum_{c \in C_u} \exitvalue{c} & \text{since $\{C_u : u \in U\}$ partitions $U \setminus \Sout$} \\
    & \le \sum_{u \in U} \frac{1}{\beta} \sum_{u \in U} \delta_u & \text{from \reflemma{deltavalue}}\\
    & =  \frac{1}{\beta} f(\Sout).
  \end{align*}
\end{proof}

The next lemma shows that $f(U)$ is not much larger than $f(\Sout)$.
\begin{lemma}
  \labellemma{all-taken-elements-upper-bound}
  $f(U) \leq \parof{1 + \beta^{-1}} f(\Sout)$.
\end{lemma}
\begin{proof}
  Let $U' = U \setminus \Sout$ and let $U' = \{v_{i_1}, \ldots,
  v_{i_h}\}$ where $i_1 < i_2 \ldots < i_h$. We have $f(U)= f(\Sout) +
  f_{\Sout}(U')$.  It suffices to upper bound $f_{\Sout}(U')$ by
  $f(\Sout)/\beta$.  For $1\le j \le h$ let $U'_j =
  \{v_{i_1},\ldots,v_{i_j}\}$.  We have $f_{\Sout}(U') = \sum_{j=1}^h
  f_{\Sout \cup U'_{j-1}}(v_{i_j})$.  We claim that $f_{\Sout \cup
    U'_{j-1}}(v_{i_j}) \le \exitvalue{v_{i_j}}$.  This follows by
  submodularity and the fact that $\Sout \cup U'_{j-1}$ is a superset
  of the vertices that are in $S$ when $v_{i_j}$ is deleted.
  Putting things together,
  \begin{align*}
    f_{\Sout}(U') = \sum_{j=1}^h f_{\Sout \cup U'_{j-1}}(v_{i_j}) \le
    \sum_{u \in U'} \exitvalue{u} \le \frac{1}{\beta}f(\Sout)
  \end{align*}
  where the last inequality follows from \reflemma{exit-values-upper-bound}.
\end{proof}

\paragraph{Relating $\opt$ to $f(U)$:} It remains to bound $f(T)$ (for some competing set $T$) to
$f(U)$ and hence to $f(\Sout)$. The critical question, addressed in the
following lemmas, is how to charge the value of elements in $T$ off to
elements in $U$.

\begin{lemma}
  \labellemma{atmost-k}
  Let $T \subseteq \vertices$ be an independent set disjoint from
  $U$. Each element $u \in U$ appears in the conflict list $C_t$ for
  at most $k$ vertices  $t \in T$.
\end{lemma}
\begin{proof}
  Fix $u \in U$.  The set $T \cap N(u) \cap \setof{v: v > u}$ consists
  of precisely the vertices $t \in T$ for which $u \in C_t$. As a
  subset of $T$, this set is certainly independent. By definition of
  $k$-inductive independence, the cardinality of this set is at most
  $k$.
\end{proof}

\begin{lemma}
  \labellemma{margin-against-all-taken}
  Let $T \subseteq \vertices$ be an independent set. Then
  \begin{align*}
    f_U(T) \leq k (1 + \beta) (1 + \beta^{-1}) f(\Sout).
  \end{align*}
\end{lemma}
\begin{proof}
  Since $f_U(T) = f_U(T \setminus U)$, it suffices to assume that $T$
  is disjoint from $U$.  For each vertex $t \in T$, since $t$ is not
  in $U$, we have
  \begin{math}
    f_{\setbefore{t}}(t) %
    \leq %
    (1+\beta) \sum_{c \in C_t} \incv{f}{\setbefore{t}}{c}.
  \end{math}
  Fix a vertex $u \in C_t$. If $u \in \Sout$, then  $u$ is in the final output;
  then we have
  \begin{math}
    \incv{f}{\setbefore{t}}{u} %
    \leq \incv{f}{\Sout}{u}
  \end{math}
  because the incremental value of an element in $S$ is
  nondecreasing. If $u \notin \Sout$, and $u$ was deleted to make room
  for some later element $u'$, then we have
  \begin{math}
    \incv{f}{\setbefore{t}}{u} %
    \leq %
    \exitvalue{u}
  \end{math}
  again because incremental values are nondecreasing.

  By the preceding lemma, each element $u \in U$ appears in $C_t$ for
  at most $k$ choices of $t$. Therefore, in sum, we have
  \begin{align*}
    f_U(T)                              %
    &\leq
      \sum_{t \in T} f_{\setbefore{t}}(t) %
    &\text{by submodularity,}\\
    &\leq                                %
      (1+\beta) \sum_{t \in T} \sum_{c \in C_t}
      \incvalue{f}{\setbefore{t}}{c} %
    &\text{since } t \notin U,
    \\
    &\leq
      k (1+\beta) \parof{\sum_{u \in \Sout} \incvalue{f}{\Sout}{u} + \sum_{u \in U \setminus \Sout} \exitvalue{u}} %
    &\text{\reflemma{atmost-k} and argument above,}
    \\
    & \leq
      k (1+\beta) \parof{f(\Sout) + \sum_{u \in U \setminus \Sout} \exitvalue{u}} %
    \\
    &\leq                       %
      k (1 + \beta) (1 + \beta^{-1}) f(\Sout) & \text{by  \reflemma{exit-values-upper-bound}}
  \end{align*}
  as desired.
\end{proof}

\medskip
From here, it is relatively straightforward to get a final
approximation bound.
\begin{theorem}
  Given an inductively $k$-independent graph with a $k$-inductive
  ordering, the algorithm \algo{preemptive-greedy} returns an
  independent set $\Sout$ such that for any independent set $T$,
  \begin{align*}
    f(T) \leq (k(1 + \beta) + 1)(1 + \beta^{-1}) f(\Sout).
  \end{align*}
\end{theorem}
\begin{proof}
  Let $T$ be an optimal solution. We have
  \begin{align*}
    f(T)                        %
    &\leq f_U(T) + f(U)    %
      \leq                 %
      (k (1 + \beta) + 1) (1 + \beta^{-1}) f(\Sout)
      \labelthisequation{preemptive-greedy-invoke-monotonicity}
  \end{align*}
  via \reflemma{margin-against-all-taken} and \reflemma{all-taken-elements-upper-bound}.
\end{proof}

\medskip
The bound is minimized by taking $\beta = \sqrt{1 + k^{-1}}$, which at
which point
\begin{align*}
  f(T) \leq (4 k + 2 + o(1)) f(\Sout),
\end{align*}
where the $o(1)$ goes to $0$ as $k$ increases. For $k = 1$, the
approximation ratio is $3 + 2 \sqrt{2}$.

\subsection{Randomized preemptive greedy for nonnegative functions}
Here we analyze the \algo{randomized-preemptive-greedy} for
non-negative submodular functions that may not be monotone.
A key observation is that the analysis of \algo{preemptive-greedy}
does not invoke the monotonicity of $f$ until the very end, in
equation \refequation{preemptive-greedy-invoke-monotonicity}. In
particular, \reflemma{margin-against-all-taken} and
\reflemma{all-taken-elements-upper-bound} hold for nonnegative
submodular functions.

\refequation{preemptive-greedy-invoke-monotonicity} invokes
monotonicity when it takes the inequality $\f{U \cup T} \geq
\f{T}$. Informally speaking, by injecting randomization, we will be
able recover a similar inequality, except losing a factor of $4$.

Fix a set $T$.  Let $\vertices'$ sample each element in $\vertices$
with probability $1/2$. Let $T' = T \cap \vertices'$. Conditional on
$\vertices'$, we have
\begin{align*}
  \f{U} \leq \parof{1 + \beta^{-1}} \f{\Sout}
\end{align*}
and
\begin{align*}
  \f_{U}{T'} \leq               %
  k \parof{1 + \beta} \parof{1 + \beta}^{-1} \f{\Sout}
\end{align*}
via \reflemma{all-taken-elements-upper-bound} and
\reflemma{margin-against-all-taken} respectively.

Now, conditional on $T'$, $U \setminus T = U \setminus T'$ is a
randomized set, where any vertex $v \in \vertices$ appears in
$U \setminus T$ with probability at most $1/2$.  By \reflemma{bfns-nonneg},
\begin{align*}
  \evof{\f{U \cup T'}}{T'} \geq \frac{1}{2} \f{T'}.
\end{align*}

We also have, via the concavity of $F$ along any non-negative
direction \cite{Vondrak-thesis},
\begin{align*}
  \evof{\f{T'}} = F(\frac12 \ones_T) \geq \frac{1}{2} F(\ones_T) = \frac{1}{2} \f(T)
\end{align*}
where $\ones_T$ is the indicator vector of $T$.

Altogether, we have
\begin{align*}
  \f{T}
  &\leq 2 \evof{\f{T'}}
  \leq
    4 \evof{\f{U \cup T'}}
    \\
  &=                             %
  4                             %
  \evof{\f{T'}_U + \f{U}}
  \leq                          %
  4 \parof{k \parof{1 + \beta} +1}\parof{1 + \beta}^{-1} \evof{\f{\Sout}},
\end{align*}
as desired.

\end{document}